\documentclass[letter,12pt]{amsart}
\usepackage{graphicx} %If you want to include postscript graphics
\DeclareGraphicsExtensions{.pdf,.eps,.fig,.pdf_tex}
\usepackage{color}
\usepackage{amsmath}
\usepackage{amssymb}
\usepackage{amsthm}
\usepackage{amscd}
\usepackage{amsfonts}
\usepackage{fancyhdr}
 \usepackage{hyperref}
 \usepackage{float}
 \usepackage{enumerate}
\usepackage[colorinlistoftodos]{todonotes}
\theoremstyle{plain} \numberwithin{equation}{section}
\newtheorem{theorem}{Theorem}[section]

\newtheorem{lemma}[theorem]{Lemma}
\newtheorem{proposition}[theorem]{Proposition}
\theoremstyle{definition}
\newtheorem{definition}[theorem]{Definition}

\newtheorem{remark}[theorem]{Remark}
\newtheorem{examplex}[theorem]{Example}
\newenvironment{example}
  {\pushQED{\qed}\examplex}
  {\popQED\endexamplex}

\title{Poly-Poisson Sigma models and their relational poly-symplectic groupoids}
\author{Ivan Contreras and Nicolas Martinez Alba}
\address{Department of Mathematics, University of Illinois at Urbana-Champaign 64801, USA}
\address{Departamento de Matem\'aticas, Universidad Nacional de Colombia, Bogot\'a, Colombia}
\thanks{}

\email{icontrer@illinois.edu, nmartineza@unal.edu.co}

\newcommand{\TM}          {T^*M}
\newcommand{\TQ}          {T^*Q}
\newcommand{\Ann}         {\mathrm{Ann}}
\newcommand{\Ad}          {{\mathrm {Ad}}}
\newcommand{\ad}          {{\mathrm {ad}}}

\newcommand{\pr}          {\mathrm{pr}}
\newcommand{\gpd}         {\mathcal{G}}
\newcommand{\F}           {\mathcal{F}}
\newcommand{\pth}         {\mathcal{P}}
\newcommand{\G}           {\mathbb{G}}

\newcommand{\LG}          {\mathfrak{g}}
\newcommand{\w}           {\omega}
\newcommand{\opr}         {\oplus_{(r)}}
\newcommand{\image}       {{\mathrm{Im}}}
\newcommand{\Ker}         {{\mathrm {ker}}}
\newcommand{\Lie}         {\mathcal L}
\newcommand{\RR}          {\mathbb{R}}

\begin{document}

\begin{abstract}
The main idea of this note is to describe the integration procedure for poly-Poisson structures, that is, to find a poly-symplectic groupoid integrating a poly-Poisson structure, in terms of topological field theories, namely via the path-space construction. This will be given in terms of the poly-Poisson sigma model $(PPSM)$ and we prove that every poly-Poisson structure has a natural integration via  relational poly-symplectic groupoids,  extending the results in \cite{CC} and \cite{Martz}. We provide familiar examples (trivial, linear, constant and symplectic) within this formulation and we give some applications of this construction regarding the classification of poly-symplectic integrations, as well as Morita equivalence of poly-Poisson manifolds.
\end{abstract}

\maketitle

\tableofcontents
\section{Introduction}

%================poly-symplectic====================\\
In the study of classical field theories, it is usually considered the variational principle of the first order Lagrangian formalism over a fiber bundle, and defined via a density (or Lagrangian) function on the first order jet manifold. The variational principle for this Lagrangian leads to the Hamilton-de Donder-Weyl equations by introducing new coordinates, denoted poly-momenta. Original developments in this direction can be traced back to the work of Carath\'eodory [4], and independentely by H. Weyl [37] and T. de Donder [8]. The Hamiltonian counterpart to the first order Lagrangian field theory on a fibre bundle is the covariant Hamiltonian formalism, where canonical momenta, or poly-momentum coordinates, correspond to jets of the field variables on a base manifold. This formalism was developed by the Polish school, based originally on the seminar by W. Tulczyjew [36] in 1968.

In this formalism there is a family of closed 2-forms defined by each poly-momentum coordinates. %see for example third equation line in page 6 in \cite{He}
Furthermore, the unique vector in the common kernel of these pre-symplectic forms is the trivial one.  The geometry underlying this formalism is known as poly-symplectic geometry \cite{Aw,Gu}, that is a smooth manifold equipped with a non degenerated vector valued closed 2-form.

%=================poly-poisson====================\\
Once we get the geometry associated to the Hamiltonian formalism of this field theory, a natural task is to extend the notion of Poisson manifold in a suitable way that includes the poly-symplectic geometry and also the equations of motion associated to the system. Among the different ways to extend the Poisson structure fitting in the classical field theory (see for example \cite{CM,FPR,Ka}) there are two definitions which induce Lie algebroid structures. In \cite{IMV}, they give a definition that extends the properties of the inverse bundle map of the one defined by the poly-symplectic form via contraction with tangent vectors. A stronger version was proposed in \cite{Martz}, and it was motivated by the use of poly-symplectic groupoids, i.e.  a Lie groupoid with a poly-symplectic structure, so that each pre-symplectic form is multiplicative, and extending the relation between symplectic groupoids and its infinitesimal counterpart, identified with Poisson structures. Moreover, the poly-Poisson structures defined in this way appear naturally in the context of Hamiltonian formalism of classical field theory as in \cite{NMA}.

In this paper we will use the second approach, that is the structure arising as the infinitesimal part of poly-symplectic groupoids. This can be identified with a sub-bundle of $r$ copies of the cotangent bundle of the base manifold, a skew-symmetric bundle map coming from the anchor, and a suitable extension of the Koszul bracket of 1-forms. This sub-bundle and the bundle map are defined via  vector valued IM-forms (c.f \cite{BC,CSS}) that is the infinitesimal geometric data of the multiplicative poly-symplectic form \cite[Sec~2]{Martz}.

%================sigma-model and rel.gpd===================\\
In the usual case, where poly-Poisson coincides with Poisson, there is another way to interpret the integration of a Poisson manifold, when it is integrable. It is done via the reduced phase space of the Poisson sigma model, a two dimensional topological field theory,  introduced by Schaller and Strobl \cite{Strobl} and independently by Ikeda \cite{Ikeda}. When the source manifold is a disk, and vanishing boundary conditions are stablished, the symplectic groupoid that integrates a given Poisson manifold is obtained via gauge reduction of the boundary fields of the theory, following the work of Cattaneo and Felder \cite{CF}. More precisely, the path-space construction of Lie groupoids integrating  Lie algebroids is incarnated as the reduced phase space of PSM, and it  is constructed from the following data:
\begin{enumerate}
\item Source data: The space-time for the PSM is given by a Riemannian surface $\Sigma$ (possibly with boundary), and a volume form $vol_{\Sigma}$. This is equivalent to give a $Q$-structure on the super manifold $T[1]\Sigma$, for which $Fun(T[1]\Sigma)= \Omega^{\bullet}(\Sigma)$.
\item Target data: The target space for the PSM is given by a Poisson manifold $(M, \Pi)$. In the language of super manifolds, it is equivalent to a $QP$-structure on $T^*[1]M$.
\end{enumerate}
 %and the target space is the cotangent bundle of a Poisson manifold $(M,\Pi)$. 

Although not all Poisson manifolds can be integrated by a symplectic groupoid, it is always possible to obtain a relational grupoid as an infinite dimensional symplectic integration. Such groupoid object in the extended symplectic category of symplectic manifolds and Lagrangian submanifolds, can be obtained by the BV-BFV formulation of the Poisson sigma model. The structure maps of a symplectic groupoid are replaced by the evolution relations obtained by the theory, which happen to be Lagrangian \cite {ConTh}.

%==================goal ===========\\
The main goal of this paper is to study the integration procedure for poly-Poisson structures in terms of an AKSZ theory, namely the poly-Poisson sigma model. In particular we have two specific objectives: to adapt the Poisson sigma model to the poly-Poisson case, and to introduce the relational groupoids arising from  this construction, as natural integrations of poly-Poisson structures. We prove that such integration always exists for integrable poly-Poisson manifolds and we characterize the construction in several examples of poly-Poisson structures. As a consequence of this construction we can also revisit three other known facts on Poisson geometry, namely the classification of integrations by relational symplectic groupoids via the path space construction, the equivalence of the integrability of a Lie algebroid in terms of the integrability of the Poisson structure of its dual, and finally the relation between Lie groupoids Morita which are equivalent and the Morita equivalence of the Poisson structures on the base manifold of the cotangent groupoid. For the case of classification, we observe in Proposition \ref{Class} that the relational poly-symplectic groupoid from the AKSZ construction is universal, i.e. other smooth integrations of the same poly-Poisson structure can be obtained via a quotient.  In the other hand, a known fact in Poisson geometry is that a Lie algebroid is integrable if and only if the Poisson structure on the dual is also integrable, so we can also wonder if such statement holds for the poly--Poisson case in lights of the direct sum of cotangent groupoid. In the case of the Morita equivalence we study the approach which says that the $C^*$-algebra of a groupoid is the deformation quantization (in Rieffel's approach) of the Poisson structure on the dual of its Lie algebroid.

%==========and organization===========
The organization of the paper is the following: Section~2 is devoted to the introduction, examples, and some properties of poly-Poisson structures. In particular we devote a subsection to the study of lagrangian and coisotropic submanifolds. We also present Proposition~\ref{prop:MW} which states the Marsden-Weinstein-Meyer reduction for poly-symplectic manifolds (for a detailed exposition see \cite{Martz}). In Section~3 we give the statement and examples of the integration of poly-Poisson structures (as Lie algebroids) to poly-symplectic groupoids. The first specific objective of the paper, the poly-Poisson sigma model, is presented in Section~4. There we give a brief summary of the Poisson sigma model, specifying the key steps to extend to the poly-Poisson case.  In particular we use an extension of the Marsden-Weinstein-Meyer reduction for the infinite dimensional weak poly-symplectic structure of the Whitney sum of cotangent path-space of the base manifold. Also, all the examples developed in the previous sections are revisited in the lights of the sigma model. Section~5 exhibits the relational groupoids in the poly-symplectic formalism. Last section is devoted to some application of the theory. First we present different integrations via relational groupoids as quotients of the one arising from the poly-Poisson sigma model. As a second consequence is the Proposition~\ref{prop:integration A vs +rA*} where we stablish integrability of a Lie algebroid in terms of the integrability of the product poly-Poisson structure of its dual. The last application is devoted to a definition of Morita equivalence for poly--Poisson manifold and the proof of Proposition~\ref{prop:poly-Wmap} that relates Morita equivalence of Lie groupoids with Morita equivalence of poly--Poisson structures.

\textbf{Acknowledgements.}The authors thank Henrique Bursztyn and Alberto Cattaneo for fruitful discussions. I.C. was partially supported by the SSNF grant P300P2-154552. N.M. thanks University of Illinois, Urbana-Champaign for the hospitality.

\section{Poly-Poisson manifolds}
We begin with the definition of the main structures to consider in this paper:
\begin{definition}\label{Def:PP}
A \textbf{Poly-Poisson  structure} of order $r$, or simply an $r$-Poisson structure, on a manifold $M$ is a pair $(S,P)$, where $S\to M$ is a vector subbundle of $\TM\otimes \RR^r$ and $P:S\to TM$ is a vector-bundle morphism (covering the  identity) such that the following conditions hold:

\begin{itemize}
\item[(i)] $i_{P(\eta)} \eta=0$, for all $\eta\in S$,
\item[(ii)] $S^\circ= \{X\in TM| i_X\eta=0, \,\forall\,\eta\in S\}=\{0\}$,
\item[(iii)] the space of section $\Gamma(S)$ is closed under the bracket
\begin{equation}\label{Def:bracketHP} 
\lfloor \eta,\gamma \rfloor:=\Lie_{P(\eta)}\gamma-i_{P(\gamma)}d\eta \; \mbox{ for } \;\gamma,\eta\in \Gamma(S),
\end{equation}
and the restriction of this bracket to $\Gamma(S)$ satisfies the Jacobi identity.
\end{itemize}
We will call the triple $(M,S,P)$ an $r$-\textbf{Poisson manifold}. 
\end{definition}

As a first remark to the previous definition is that for the case $r=1$ we recover the usual notion of Poisson manifold. This comes just by noticing that $S=\TM$ (by item (ii)) and $\pi^\sharp:=P$ is a bivector (by item (i)), moreover the condition on the bracket on 1-forms says that $\pi$ is indeed a Poisson bivector.

Recall that Poisson structures are equivalently defined via a bracket on the algebra of smooth functions of the manifold. For the case of poly--Poisson structures we also can define a bracket operation but on the space of admissible functions. For this we must define the space of admissible functions as
$$C_{\mathrm{adm}}^\infty(M,S)=\{h\in C^\infty(M,\RR^r):  dh\in S\}$$
and the bracket between two admissible function is 
\begin{equation}\label{eq:bracket on functions}
\{h,g\}=\Lie_{P(dh)}g=(\Lie_{P(dh)}g_1,\ldots,\Lie_{P(dh)}g_r).
\end{equation}
As expected, it is possible to verify that the bracket satisfies the Jacobi identity and $\{h,fg\}=\{h,f\}g+f\{h,g\}$ for $f,g$ two admissible function do that $fg\in C_{\mathrm{adm}}^\infty(M,S)$. Despite these facts, the bracket on admissible functions does not define the poly--Poisson structure $(M,S,P)$, as will be shown in Remark~\ref{rmk:adm-brck no PP}. 

\subsection{Examples}\label{sec:examples} 
In this part we will present the basic examples of these structures. For a complementary and more detailed exposition see e.g. \cite{Martz}.

\begin{example}[Poly-symplectic structures]\label{ex:PS}
First we will suppose the case of $P$ is an isomorphism of vector bundles. By using the natural projections $p_j:S\to \TM$ we can define the following bundle map $$\w_j:TM\overset{P^{-1}}{\longrightarrow}S\to \TM.$$ Condition (i) in Def.~\ref{Def:PP} is the same as that each $\w_j$ is skew-symmetric, whereas condition (ii) means that $\cap \Ker \w_j=0$. Finally, condition (iii) is equivalent to $d\w_j=0$ for $j=1,\ldots,r$. 
%==========Proof of the claims:===========
%(i) impplies $\w_j(X)(X)=i_{P(P^{-1}(X))}P^{-1}(X)=0$; (ii) says that for $Y\in \cap \Ker \w_j$, then $i_X\w_j(Y)=0$ for each $j$, which means that $X\in S^o=0$; $d\w_j=0$  is just a rewriting of (iii)
%========================================
What we get here is that $(M,\w_1,\ldots,\w_r)$ is an $r$-{\em poly-symplectic} manifold, or $r$-{\em symplectic}, as in \cite{Gu}. In the same way, any poly-symplectic manifold $(M,\w)$ with $\w=(\w_1,\ldots,\w_r)$ induces an $r$-Poisson manifold with $S_\w=\image (\w^\sharp)$ and $P_\w=(\w^\sharp)^{-1}$.
%\todo[inline]{which is this ´Weak poly-symplectic structures´?}

\end{example} 

The canonical example of this kind of structure is {\em the space of covelocities} $\opr \TM$ for a fixed manifold $M$. For this purpose we consider the canonical symplectic form $\w_{can}$ on $\TM$ and define 
\begin{equation}\label{eq:can.PS}
\w=(p_1^*\w_{can},\dots,p_r^*\w_{can})
\end{equation}

for $p_j:\opr \TM\to \TM$ the $j$th-projection.

\begin{remark}\label{rmk:adm-brck no PP}
Now we can construct a poly--Poisson structure whose bracket on admissible functions  does not define the structure. For this consider $(M_1\times M_2,S_\w,P_\w)$ as before with $\w_1$ a  symplectic form on $M_1$ and $\w_2$ closed 2-form with non-trivial kernel on $M_2$. In this situation, the bracket does not defines the vector field $P(dh)$ for an admissible function $h$ because $\{h,\cdot\}$ is not a derivation of $C^\infty(M_1\times M_2)$. 
\end{remark}

\begin{example}[Trivial structure]\label{ex:trivial}
Let $Q$ be a manifold. For each natural number $r$ we can view $Q$ as an $r$-Poisson manifold, and this can be done in several ways. For example, $S_1=\TQ\otimes \RR^r$ and $P_1=0$ define a Poly-Poisson structure on $Q$. The same is true by taking a collection $(\zeta_1,\dots, \zeta_r)$ of non degenerate $1$-forms on $Q$ and defining $S_2=\{(c_1\zeta_1,\dots, c_r\zeta_r):c_j\in \RR\}$ and $P_2=0$. Other examples are
\begin{center}\begin{tabular}{ll}
$S_3=\{\alpha\oplus\ldots\oplus\alpha\,|\, \alpha \in \TQ\}
\subset \opr  \TQ$& and $P_3=0$, \\
$S_4=\{\alpha\oplus0\oplus\ldots\oplus 0\,|\, \alpha \in \TQ\}
\subset \opr \TQ$ &and $P_4=0$.
\end{tabular}\end{center}

\end{example} 

\begin{example}[Product of Poisson structures]\label{ex:PPprod}
Consider $(M_i,S_i,P_i)_{i=1,\ldots,k}$ a family of $r_i$-Poisson manifolds and denote $M=M_1\times \cdots M_r$. Let $\bar{S}_j$ be the natural inclusion of $S_j$ into  $\TM$. %i.e. the inclusion by $pr_j^*$
Let $S\subset \TM\otimes \RR^r$ and $P:S\to TM$ be defined by 
\begin{equation}\label{eq:S in product}
\begin{cases}
S= \bar{S}_1\oplus \ldots \oplus \bar{S}_l \\ 
P(\alpha_1,\ldots,\alpha_l)=(P_1(\alpha_1),\ldots,P_l(\alpha_l)).
\end{cases}
\end{equation}
%$S=\{(pr_1^*\alpha_1,\dots,pr_r^*\alpha_r):\alpha_j\in S_j\}$
%is well defined bcs the inclusion $\bar{S}_j$ is given by the pull-back of the projection map, which is injective
One may verify that $(M,S,P)$ is an $r_1+\cdots r_k$-Poisson manifold directly from the definition. 
\end{example} 

Note that this construction says that the direct product of $r$ Poisson structures is an $r$-Poisson structure. And when $(S_i,P_i)$ comes as the image of a $r_j$-symplectic structure, then $(S,P)$ is indeed $r_1+\cdots r_k$-symplectic with the vector valued 2-form 
$$\w=(p_1^*\w_1,\dots,p_r^*\w_r)$$
where $p_i:M\to M_i$ is the canonical projection.

\begin{example}[Constant poly-Poisson structure]\label{ex:constant}
The constant Poisson structure can be seen locally as the product of a symplectic and a trivial Poisson structure. We can use the same idea to define constant poly-Poisson structure on $\RR^k\times \RR^m$ for $(\RR^m,\w)$ an $s$-symplectic manifold. The constant $r+s$-Poisson is $S=(T^* \RR^k\otimes \RR^r)\oplus S_\w$ with $P=0\oplus P_\w$.

\end{example} 

\begin{example}[Linear poly-Poisson structures]\label{ex:linear}
Let $\LG$ be a Lie algebra, and let
$$
\LG_{(r)}:=\LG\times \overset{(r)}{\cdots} \times \LG,\;\;\;
\LG^*_{(r)}:=\LG^*\times \overset{(r)}{\cdots} \times \LG^*.
$$
For $u\in \LG$, let $u_j \in \LG_{(r)}$ denote the element
$(0,\ldots,0,u,0,\ldots,0)$, with $u$ in the $j$-th entry.
 Since $\LG^*$ is equipped with its
Lie-Poisson structure, $\LG^*_{(r)}$ naturally carries a product $r$-Poisson structure, as in Example~\ref{ex:PPprod}. More
important to us is the following {\it direct-sum} poly-Poisson
structure (see \cite{IMV} for the structure and \cite{Martz} for the construction of direct sum $r$-Poisson structure): over each
$\zeta=(\zeta_1,\dots ,\zeta_r)\in \LG^*_{(r)}$, we define
\begin{equation}\label{eq:bundle-r-linear}
S|_\zeta:=\{(u_1,\dots ,u_r)|u\in \LG\} \subseteq \opr T^*_\zeta
\LG^*_{(r)} \cong \opr \LG_{(r)},
\end{equation}

and the bundle map $P: S \to T \LG^*_{(r)}$,

\begin{equation}\label{eq:anchor-r-linear}
P_\zeta(u_1,\dots ,u_r):=(\pi_{\zeta_1}^\sharp(u),\ldots,\pi_{\zeta_r}^\sharp(u))=(\ad_{u}^*\zeta_1,\dots ,\ad_{u}^*\zeta_r)
\in T_\zeta \LG^*_{(r)}\cong \LG^*_{(r)}.
\end{equation}

%The diagonal coadjoint action shows that
%$\bar{\eta}P(\bar{\eta})=0$ for each $\bar{\eta}\in S_\zeta$. For
%the non-degenaracy condition tare $X\in T_\zeta \LG_{(r)}^*$, then
%we get that $X\in S_\zeta ^o$ if and only if $\langle X_j,u
%\rangle=0$ for each $u\in \LG$ and $j=1,\dots ,r$, this implies
%$S_\zeta^o=\{0\}$.

\end{example}

\subsection{Some special submanifolds}

This part is devoted to the extension of the classical definition of symplectic foliation, coisotropic and Lagrangian submanifolds to the context of poly-symplectic and poly-Poisson. The more important remark in this subsection is that we recover, as in the usual case, the reduction of coisotropic submanifolds.

\subsubsection{Poly-symplectic foliation}
A direct consequence of Definition.~\ref{Def:PP} is that $S$ has a Lie algebroid structure with anchor $P$, hence the distribution $D := P(S) \subset TM$ is integrable, and its leaves define a singular foliation on $M$. Each leaf $\iota: \mathcal{O}\hookrightarrow M$ carries an $\mathbb{R}^r$-valued $2$-form $\w_\mathcal{O}$ determined by the condition
\begin{equation}\label{eq:H-Fol}
\w_\mathcal{O}^\sharp : T_m\mathcal{O} \to 
T_m^*\mathcal{O}\otimes \RR^r,\;\;\; P(\eta) \mapsto \iota^*\eta.
\end{equation}
The fact that the 2-form $\omega_\mathcal{O}$ on the leaf $\mathcal{O}$ is well defined follows from (i) in Definition \ref{Def:PP}, while (iii) guarantees that it is closed, and , (ii) guarantees that it is non-degenerate. Summarizing, $(M,S,P)$ determines a singular foliation on $M$ with poly-symplectic leaves.

A first remark on the poly-symplectic foliation of a $r$-Poisson structure is that, in contrast with the case $r=1$, different $r$-Poisson structures may correspond to the same poly-symplectic foliation, as shown in the next example.
\begin{example}\label{ex:PSfoliation is not unique}
%A symplectic foliation on a manifold determines uniquely the Poisson structure on the manifold, but in the case of the poly-symplectic  foliations the assertion on the uniqueness of the poly-Poisson  structure does not hold.
Let $\w_s$ be a smooth family of $r$-poly-symplectic forms on $M$ parametrized by $s\in \mathbb{R}$ and define the following vector subbundles of $T^*(M\times \mathbb{R})\otimes \RR^r$:%in particular we can take $\w_s=\w$ for $w$ poly-symplectic
\begin{center}\begin{tabular}{l}
$S_1\big{|}_{(m,s)}:=\{(i_X\w_s,\gamma_1\oplus \cdots \oplus \gamma_r)|X\in T_mM, \gamma_j\in T_s^*\mathbb{R}\},$\\
$S_2\big{|}_{(m,s)}:=\{(i_X\w_s,\gamma \oplus \cdots \oplus \gamma)|X\in T_mM, \gamma\in T_s^*\mathbb{R}\},$\\
$S_3\big{|}_{(m,s)}:=\{(i_X\w_s,\gamma\oplus 0 \cdots \oplus 0)|X\in T_mM,
\gamma\in T_s^*\mathbb{R}\}.$
\end{tabular}\end{center}
On each $S_j$ we define $P_j(i_X\w_s,\bar{\gamma})=X.$ Each $(S_j,P_j)$ is a poly-Poisson structure on $M\times \mathbb{R}$. Note that by definition of $(S_j,P_j)$ we obtain that $TM=P_j(S_j)$ and the leaves associated to this distribution are $M\times \{s\}$ for each $s\in \RR$.  
%===============poly-Poisson==============================
%The condition $S_j^o=\{0\}$ is a direct verification from $\ker(\w_s)=\{0\}$. The involution w.r.t. the bracket follows the ideas that for $\alpha,\beta\in \Gamma(S)$ we have 
%$$\Lie_{P(\alpha)}\beta|_{(m,s)}=(\Lie_Xi_Y\w_s)|_m \mbox{\ and\ }i_{P(\beta)}d\alpha|_{(m,s)}=(i_{Y}di_X\w_s)|_m,$$ then $\lfloor \alpha,\beta\rfloor|_{(m,s)}=(i_{[P(\alpha),P(\beta)]}\w_s,0)\in L|_{(m,s)}$. 
%============================================================
%as on defintion \ref{Def:PP} and $(S_1,P_1)$ is the structure defined on proposition \ref{Prop:PSfoliation}, but these three $k$-poly-Poisson structures have the same poly-symplectic foliation on $M\times \mathbb{R}$.
%===============weak-poly-Poisson not poly-Poisson==============================
%Same conclusion holds for the weak-poly-Poisson structure given by
%$$S_0\big{|}_{(m,t)}:=\{(i_X\w_t,0)|X\in T_mM\} \mbox{\ \ and\ \ }P_0(i_X\w_t,0)=X$$
%where the poly-symplectic foliation is described on Theorem 3.4 on \cite{IMV}.
%===========================================================================
\end{example} 

\begin{example}\label{ex:productofcoadjoint}
In the cases of product of linear Poisson (as in Example~\ref{ex:PPprod}) and the linear $r$-Poisson (as in Example~\ref{ex:linear}) we obtain the same foliation which is the product of the coadjoint orbits.
\end{example}

\subsubsection{Lagrangian and coisotropic submanifolds}

For an $r$-symplectic manifold $(M,\w_1,\ldots,\w_r)$ we say that a subbundle $L\leq TM$ is Lagrangian (resp. isotropic or coisotropic) if $L^\w=L$ (resp. $L\leq L^\w$ or $L^\w\leq L$) where $$L^\w=\{X\in TM:i_Xi_Y\w_i=0 \mbox{ for all\ } Y\in L, i=1,\ldots,r\}.$$ 
In the symplectic case there is an equivalent statement for the Lagrangian subspaces: $L$ is Lagrangian if and only if  the linear map $\w^\sharp:L\to Ann(L)$, defined by contraction, is an isomorphism. For a poly-symplectic manifold $(M,\w_1,\ldots,\w_r)$ we can also define the induced map $\w^\sharp$ for which we get that
\begin{proposition}\label{prop:pL-L}
Let $(M,\w)$ be a $r$-symplectic manifold  and $L\leq TM$ a subbundle. If $\w^\sharp:L\to Ann(L)\otimes \RR^r$ is an isomorphism then $L^\w=L$.
\end{proposition}
In the literature (see \cite{Aw,FoGo}) first statement is known as poly-Lagrangian subspace and the existence of such subbundle is equivalent to the existence of a Darboux coordinate system for poly-symplectic manifolds. In contrast to the usual symplectic manifolds, there exists poly-symplectic manifolds that do not allow Darboux coordinates:
\begin{example}
Let $M$ be an oriented surface with volume form $\nu$. Consider $\eta$ any 2-form in $M$, then $(M,\nu,\eta)$ is a 2-symplectic manifold. Note that the dimensional identity $\dim L=2\dim L$ is not satisfy for all subspace $L\leq V$, hence there is no poly-Lagrangian subbundle of $M$\footnote{Following conditions for poly-symplectic Darboux theorem in \cite{Aw} we can conclude that for $M$ an orientable surface with volume form $\nu$ the poly-symplectic manifold $(M,\nu,\eta)$ has not Darboux coordinates}. 
%============another example of lagrangian ~poly-lagrangian==================
%By using this construction we can note that $M=\RR^2$ has a Lagrangian subspace which is not poly-Lagrangian. Just take $L\leq T\RR^2$ Lagrangian in $(M,\nu=\w_{can})$ and $\eta=0$.
%========================================================================
\end{example}
\begin{proof}[of Proposition~\ref{prop:pL-L}]
Let $L$ be as in the proposition. As the image of the map $\w^\sharp$ is $Ann(L)\otimes \RR^r$ we obtain that $L$ is isotropic. To prove that is also coisotropic we fix $X\in L^\w$, then $i_X\w\in Ann(L)\otimes \RR^r$ and the isomorphism gives us that $X\in L$. 
\end{proof}
In contrast to the symplectic case, both statements are not equivalent as is shown in the following examples.

\begin{example}
Consider $M=\RR^3$ with poly-symplectic form $\w=(dx_1\wedge dx_2,dx_2\wedge dx_3)$. It is easy to verify that $L=span\{\partial x_2\}$ is Lagrangian but, by a dimensional argument, it is not poly-Lagrangian.
\end{example}

In the Poisson case we have two (equivalent) ways to define coisotropic subspaces. If $(M,\pi)$ is Poisson manifold, we say that $L\leq V$ is coisotropic if $\pi^{\sharp}(Ann(L))\subset L$ or equivalently  $Ann(L)\subset (Ann(L))^\perp$ where 
$$W^\perp:=\{\beta\in V^*:i_{\pi^ \sharp(\alpha)}\beta=0 \mbox{\ for all\ }\alpha\in W\}.$$ Both conditions also can be defined in the setting of a $r$-Poisson manifold $(M,S,P)$ with suitable changes as follows:
\begin{center}\begin{tabular}{lr}
(a)$P(S\cap Ann(L)\otimes \RR^r)\subset L$ & (b)$Ann(L)\otimes \RR^r\subset (S\cap Ann(L)\otimes \RR^r)^\perp$.
\end{tabular}\end{center}

In this case we again obtain that they are equivalent statements (see Lem.~2.3.18 \cite{NMA}). Any $L\leq V$ satisfying one of the previous relation is called poly-Poisson coisotropic subbundle. 

\begin{lemma}
For an $r$-symplectic manifold $(M,\w)$ we have that a subbundle $L\leq TM$ satisfies $L^\w\leq L$ if and only if $L$ is poly-Poisson coisotropic.
\end{lemma}
\begin{proof}
$(\Rightarrow)$ Uses condition (b) and $(\Leftarrow)$ uses (a).
\end{proof}

\begin{remark}
In \cite{Martz} it is defined the notion of poly-Poisson morphism between two $r$-Poisson manifold. Under the definition of coisotropicity we get that a map $f:(M,S_M,P_M)\to (N,S_N,P_N)$ is poly-Poisson morphism if and only if $f^*S_N\subset S_M$ and $Q=graph(f)$ is a coisotropic submanifold of $\subset M\times \bar{N}.$
\end{remark}

In the case of a submanifold $Q$ of a poly-symplectic manifold $(M,\w_1,\ldots,\w_r)$, this definition is equivalent to the fact that for any point $q\in Q$ the vector space $T_qQ$ is coisotropic in $T_qM$. As usual, we can verify that (possible non constant rank vector bundle) $TQ^\w$ is involutive, hence it induces a (possible singular) isotropic foliation. Following the same ideas in Lemma~2.7 and Lemma~5.35 \cite{McD} we obtain a coisotropic reduction.
\begin{proposition}
If $Q$ is a regular coisotropic submanifold of the poly-symplectic manifold $(M,\w)$ then the quotient $M' = Q/ Q^{\w}$ is a poly-symplectic manifold.
\end{proposition}

\subsection{Reduction by symmetries}
In this section we will give the basic idea for the reduction of the structure and in particular state the Marsden-Weinstein reduction for poly-symplectic manifolds. For this, we will assume that there is a (free and proper) action of a Lie group $\G$ on a poly-symplectic manifold $(M,\w)$ \footnote{We will see that at the infinite dimensional level, i.e. Banach manifolds with poly-symplectic structures, it is enough to have a Lie algebra action.}, that is $\varphi_g^*\w=\w$ for any $g\in \G$  where $\varphi:\G\times M\to M$ is the action of the Lie group $\G$ on $M$. An action is called reducible if
\begin{equation}\label{Def:p.reductible}
\begin{cases}
(a)\ S\cap \oplus_r \Ann(V) \mbox{\ \ has constant rank.\ } \\
(b)\ (S\cap \oplus_r \Ann(V))^\circ\subset V,
\end{cases}
\end{equation}

where $V\subseteq TM$ denotes the vertical bundle defined by this, or equivalently, is the kernel of the differential of the projection map $\Pi:M\to M/\G$.

In the case of a $r$-symplectic manifold equipped with a poly-symplectic action  This action is called {\em
hamiltonian} \cite{Gu,MR-RSV} if there is a {\em moment map}, i.e., a map $J:M\to
\LG^*_{(r)}$ that satisfies
%\begin{equation}\label{Def:HamPPaction}
\begin{center}\begin{tabular}{rlclr}
(i)&$J\circ \varphi_g=\Ad_g^*\circ J$& and &(ii)&$i_{u_M}\w=d\langle
J,u\rangle$.\\
\end{tabular}\end{center}
%\end{equation}
for all $u\in \LG$. Here $Ad_ g^*$ denotes the diagonal coadjoint action on $\LG^*_{(r)}$, and $u_M\in \mathfrak{X}(M)$ is the infinitesimal generator corresponding to $u\in \LG$. 

As a direct consequence of the previous definitions, it is possible to prove that 
$J^*(u_1, \dots, u_r)\in S_w $ for any $(u_1, \dots, u_r)\in S$ and 
$$P(u_1, \dots, u_r)=dJ(P_\w(J^*(u_1, \dots, u_r)))$$
where $(S_\w,P_\w)$ and $(S,P)$ are the $r$-Poisson structures in  Example~\ref{ex:PS} and Example~\ref{ex:linear} respectively. These two facts generalize, to the $r$-Poisson case, the well known fact that moment map of Hamiltonian action is a Poisson morphism (see Proposition~4.2 in \cite{Martz}).

A \textit{clean value} for $J$ is an element $\zeta\in \LG^*_{(r)}$ so that
\begin{equation}\label{Def:cleanvalue}
\begin{cases}
 J^{-1}(\zeta) \text{ \ is a submanifold of\ }M,\\
 \ker(d_xJ)=T_xJ^{-1}(\zeta), \text{ \ for all \ } x\in  J^{-1}(\zeta).
\end{cases}
\end{equation}
%\comment{from now on I will use clean value. Regular value as particular case}
The submanifold $J^{-1}(\zeta)$ is invariant by the action of the isotropy group $\G_\zeta$. We assume that the $\G_\zeta$-action on $J^{-1}(\zeta)$ is free and proper, so we can consider the reduced manifold
$$
M_\zeta:=J^{-1}(\zeta)/\G_{\zeta}.
$$

This last condition, together with the $\G_\zeta$-invariance of $i_\zeta^*\omega$, implies that $i_\zeta^*\omega$ is basic, i.e., there exists a (unique) closed form $\w_{red}\in \Omega^2(M_\zeta, \mathbb{R}^r)$ so that
\begin{equation}\label{eq:wr}
\Pi_\zeta^*\w_{red}=i_\zeta^*\w.
\end{equation}

Note that the second condition on the definition of Hamiltonian action says that
$$
TJ^{-1}(\zeta) = \ker(dJ) = V^\omega.
$$
 One may also check that

$$
(\ker(dJ))^\omega = (V^\omega)^\omega = (S_\w\cap \oplus_r
\Ann(V))^\circ,
$$
and finally obtain that
$$
\ker(i_\zeta^*\omega) = (TJ^{-1}(\zeta))^\omega\cap TJ^{-1}(\zeta) =
(S_\w\cap \oplus_r \Ann(V))^\circ \cap T J^{-1}(\zeta).
$$
This leads us to conclude the following condition for the poly-symplectic Marsden-Weinstein reduction (for details we refer to Section~4.2 in \cite{Martz}):

\begin{proposition}\label{prop:MW}
The reduced form $\omega_{red} \in \Omega^2(M_\zeta, \mathbb{R}^ r)$
defined by \eqref{eq:wr} is $r$-symplectic if and only if
\begin{equation}\label{eq:redpoly}
(S\cap \oplus_r \Ann(V))^\circ\cap TJ^{-1}(\zeta) \subseteq V_\zeta
= V\cap TJ^{-1}(\zeta).
\end{equation}
\end{proposition}

Note that for a reducible action this reduction condition holds.

\begin{example}[Reduction of the spaces of covelocities]\label{ex:covelocities reduction}
Let $Q$ be a manifold equipped with a free and proper $\G$-action, and let $(M=\oplus_{(r)}T^*Q, \w)$  be the $r$-symplectic manifold of covelocities. We keep the notation $\pr_j:M\to T^*Q$ for the natural projection onto the $j$th-factor. The cotangent lift of the $\G$-action on $Q$ defines an action on $T^*Q$, which induces a $\G$-action on $M$ preserving the poly-symplectic structure (i.e., it is a poly-Poisson action). The first claim is that  
%and there is a natural identification
%$$
%M/\G \cong \oplus_{(k)} (T^*Q/\G).
%$$
%We observe here that both conditions in \eqref{Def:p.reductible}
%hold, i.e., 
the $\G$-action on $M$ is reducible. To verify this fact, let $V\subseteq TM$ be the vertical bundle of the $\G$-action on $M$, so that $V_j=d\pr_j(V) \subseteq T(T^*Q)$ is the vertical bundle of the $\G$-action on the $j$th-factor $T^*Q$. Note that the natural projection $T^*Q\to Q$ induces a projection of $V_j^{\w_{can}}$ onto $TQ$, and one then sees that
$$
V_1^{\w_{can}} \times_{TQ}\ldots \times_{TQ} V_r^{\w_{can}}\subseteq
T(T^*Q) \times_{TQ}\ldots \times_{TQ} T(T^*Q)= TM
$$
is a vector subbundle, that we denote by $W$. One can now check that
\begin{equation}\label{eq:Scovel}
S_\w \cap \opr \Ann(V)=\{i_X\w\,|\, X\in W\},
\end{equation}
from where one concludes that condition (a) of \eqref{Def:p.reductible} holds. From \eqref{eq:Scovel}, one directly sees that
\begin{align*}
(S_\w\cap \opr \Ann(V))^\circ &= (V_1^{\w_{can}})^{\w_{can}}
\times_{TQ}\ldots \times_{TQ} (V_r^{\w_{can}})^{\w_{can}}\\
& = V_1 \times_{TQ}\ldots \times_{TQ} V_r = V,
\end{align*}
showing that (b) of \eqref{Def:p.reductible} also holds. So the action is reducible. In the same way as we induce the reducible action on $M$ we can induce a Hamiltonian action with moment map $J:M\to \LG^*_{(r)}$ defined by  
$$J(m,\alpha)(u):=\alpha_m(u_M(m)).$$
As the action is reducible and is Hamiltonian we obtain $r$-symplectic Marsden-Weinstein reduction at level 0.

\end{example} 

\section{Poly-Poisson structures and their integration}

It is well known that the infinitesimal counterpart of symplectic groupoids are Poisson manifolds. The aim of this section is to rephrase this in the context of poly-symplectic Lie groupoids. Under such definition, we will comment about its infinitesimal data which will coincides with the poly-Poisson manifolds. 

For  a Lie groupoid $\mathcal{G}$ over $M$ and a multiplicative form $\theta\in \Omega^k(\mathcal{G})$ is a differential form satisfying the relation
\begin{equation}\label{eq:mult-form}
m^*\theta = \pr_1^*\theta + \pr_2^*\theta
\end{equation}
for the natural projection maps $\pr_i:\gpd_{(2)}\to \gpd$ and the partial product $m:\gpd_{(2)}\to \gpd$. Observe that it still makes sense for $\mathbb{R}^r$-valued forms $\theta=(\theta_1,\ldots,\theta_r)$, which it simply says that each component $\theta_i$ is multiplicative. 

As a direct generalization of the notion of \textit{ symplectic groupoid}, see e.g. \cite{CDW,We}, we are led to the definition of the main structure in the section:

\begin{definition}\label{def:kr-symp. gpd.}
A $r$-\textbf{symplectic  groupoid} is a Lie groupoid $\mathcal{G}\rightrightarrows M$ together with a $r$-symplectic form $\omega= (\omega_1,\ldots,\omega_r) \in \Omega^{2}(\mathcal{G},\mathbb{R}^r)$ satisfying \eqref{eq:mult-form}. More explicitly, each $\omega_j\in \Omega^{2}(\mathcal{G})$ is closed, multiplicative, and $\cap_{j=1}^r \Ker(\omega_j)=\{0\}$.
\end{definition}

Here we present the structure of the product of symplectic groupoids (see Proposition 2.2.2~\cite{NMA})

\begin{proposition}\label{prop:prod}
The direct product of symplectic groupoids
$(\mathcal{G}_j,\omega_j)_{j=1,\ldots,r}$, naturally carries a
multiplicative poly-symplectic structure given by
$$
\omega=(p_1^*\omega_1,\ldots,p_r^*\omega_r),
$$
where $p_j: \mathcal{G}_1\times\ldots\times \mathcal{G}_r\to
\mathcal{G}_j$ is the natural projection.
\end{proposition}

\begin{example}\label{ex:kr-cotangent gpd}
The symplectic manifold $\TQ$ is a symplectic groupoid over $Q$, with respect to fibrewise addition; the source and target maps coincide with the projection $\TQ\to Q$. A direct consequence of the previous constructions shows that the vector bundle over $Q$,
$$\TQ\otimes \RR^{r}\simeq \opr\TQ=\TQ \oplus \cdots \oplus  \TQ$$
is a Lie groupoid with objects manifold $Q$ and is endowed with $(p_1^*\w_{can},\dots ,p_r^*\w_{can})$ multiplicative $r$-symplectic form making $\opr  \TQ \to Q$ a $r$-symplectic groupoid.

\end{example} 

The main result of this subject is the following integration theorem, for the proof see \cite{Martz}.

\begin{theorem}[Integration of poly-Poisson structures]\label{Thm:integration}
If $(\mathcal{G}\rightrightarrows M,\w)$ is a $r$-symplectic groupoid, then there exists a unique $r$-Poisson  structure $(S,P)$ on $M$ such that $S=\image(\mu)$ while $P$ is determined by the fact that the target map $t:\mathcal{G}\to M$ is a higher-Poisson morphism.
 
Conversely, let $(M,S,P)$ be a $r$-Poisson  manifold and $\mathcal{G}\rightrightarrows M$ be a $s$-simply-connected groupoid integrating the Lie algebroid $S\to M$. Then there is a $\w\in \Omega^{2}(\mathcal{G},\mathbb{R}^r)$, unique up to isomorphism, making $\mathcal{G}$ into a $r$-symplectic groupoid for which $t:\mathcal{G}\to M$ is a higher-Poisson morphism.
\end{theorem}

We say that a poly-symplectic groupoid integrates a poly-Poisson structure if they are related as in the statement of  Theorem \ref{Thm:integration}. It is easy to observe that this correspondence between source-simply-connected poly-symplectic groupoids and poly-Poisson manifolds (with integrable Lie algebroid) extends the well-known relationship between (source-simply connected) symplectic groupoids and Poisson manifolds when $r=1$.

\begin{example}[Integrating the poly-symplectic structure]\label{ex:integration-PS}
The poly-Poisson structure in Example~\ref{ex:PS} is integrated by the pair groupoid $M\times M\rightrightarrows M$, equipped with the higher-symplectic  structure $t^*\omega - s^*\omega \in \Omega^{2}(M\times M, \mathbb{R}^r)$, where $s,t$ are the source and target maps on the pair groupoid, i.e $t(x,y)=x$ and $s(x,y)=y$.
In the case that $M$ is not simply-connected, the $s$-simply connected Lie groupoid integrating $(M,\w)$ is the fundamental groupoid $\Pi(M)$ with higher-simplectic form given by the pull-back  $F^*(t^*\omega - s^*\omega)$ by the covering map $F:\Pi(M)\to M\times M$ which is groupoid morphism.

\end{example}

\begin{example}[Integrating the trivial structure]\label{ex:integration-trivial}
The integration of these poly-Poisson structures $(S_j,P_j)_{j=1,2,3,4}$ are $S_j$ itself, viewed as a vector bundle over $Q$ with groupoid structure given by fibrewise addition and equipped $r$-symplectic form the pull-back of the canonical multisymplectic form $\w_{can}$ of $\TQ\otimes \RR^r$ in Example~\ref{ex:kr-cotangent gpd}.

\end{example}

The previous description of the possible types of trivial poly-Poisson structure and its integration leads us to define the trivial $r$-Poisson structure over a manifold $Q$ as $S=\TQ\otimes \RR^r$ with $P=0$ because is the integration of the $r$-symplectic version of the cotangent bundle.
%this claim is true because all the higher-Poisson $S_j$ defined on the example are 
%Lie sub-algebroids of $(S,P)$. As Lie sub-algebroids of integrable Lie algebroid are 
%integrables too then they are integrables and the inclusion is groupoid morphism preserving
%the form by pull-back.

\begin{example}[Integrating the product structure]\label{ex:integration-product}
Recall the poly-Poisson product in Example~\ref{ex:PPprod}. Suppose that all the Poisson manifold are integrable with symplectic groupoid $(\gpd_j\to M_j,\w_j)$. The poly-Poisson  $(S,P)$ in Example~\ref{ex:PPprod} is integrable and its integrating poly-symplectic groupoid is given by the direct product Lie groupoid $$\gpd=\gpd_1\times \ldots\times \gpd_l\rightrightarrows M.$$

\end{example}

\begin{example}[Integrating constant poly-Poisson structures]\label{ex:integration-constant}
In the usual Poisson case we get constant Poisson structure on $\RR^k\times \RR^{2n}$ as the infinitesimal version of the symplectic groupoid $\gpd=T^* \RR^k\times (\RR^{2n}\times \RR^{2n})$ with multiplicative symplectic form $\w_{can}+ \pr_1^*\w_{can}-\pr_2^*\w_{can}$ where $\w_{can}$ is the canonical symplectic form in the respective manifold. For the constant structure in Example~\ref{ex:constant}, the integrating $r+s$-symplectic groupoid is $(T^* \RR^k\otimes \RR^r)\times \RR^m\times \RR^m\rightrightarrows \RR^{k+m},(\w_{can},\pr_1^*\w-\pr_2^*\w)$ following the construction in the previous example.

\end{example}

\begin{example}[Integrating the linear case]\label{ex:integration-linear}

Recall that for a Lie algebra $\LG$ the space $\LG^*_{(r)}$ has associated two $r$-Poisson structures, the product and the direct-sum (see Example~\ref{ex:linear}). The first one is integrated by the direct product of $r$ copies of the symplectic groupoid $(T^*\G,\w_{can)}$. In this case we over each $\zeta=(\zeta_1,\dots ,\zeta_r)\in \LG^*_{(r)}$, we have 
$$S_\zeta=T_{\zeta_1}^*\LG^*T_{\zeta_1}^*\LG^*\overset{r}{\cdots}\oplus T_{\zeta_r}^*\LG^*$$
and the bundle map $P: S \to T \LG^*_{(r)}$ is $P_\zeta(u_1,\ldots,u_r)=(\pi_{KKS}^\sharp|_{\zeta_1}(u_1),\ldots,\pi_{KKS}^\sharp|_{\zeta_r}(u_r))$.

The second way its by the consideration of the Lie groupoid action $\G\ltimes \LG^*_{(r)}$ with structural maps induced by the source and target 
\begin{center}\begin{tabular}{cc}
$s(g,\zeta)=\zeta,$&$t(g,\zeta)=\Ad_g^*\zeta$.\\
\end{tabular}\end{center}
Using the identification $T^*\G \cong \G \times\LG^*$ (by right translation), we see that
$$  T^*\G \otimes \RR^r \cong \G\times \LG_{(r)}^*,$$
so we may consider $T^*\G \otimes \RR^r$ as a Lie groupoid, and its canonical poly-symplectic structure $\w_{can}$ makes it into a poly-symplectic groupoid.

\end{example} 

It is worth to compare some facts of symplectic groupoids in the poly-symplectic case. Recall that  for symplectic groupoids $(\gpd\rightrightarrows M,\w)$ we have that the unit manifold is Lagrangian and that the graph of the multiplication is also Lagrangian in $\gpd\times \gpd\times \bar{\gpd}$. In the poly-symplectic case we have the following result and counterexamples regarding the relationship between ploy-symplectic structures and Lagrangian structures.

As we have seen in Proposition~\ref{prop:pL-L}, the fact that $gr(m)$ is Lagrangian is weaker than poly-Lagrangian, here we show that for any symplectic groupoid $(\gpd\rightrightarrows M,\w)$ we can construct a 2-symplectic groupoid $(\gpd\rightrightarrows M,\w,\eta)$ that is not poly-Lagrangian. For this just take $\eta=t^*\eta_0-s^*\eta_0$ for $\eta_0\in \Omega^2(M)$ and verify that $(\w,\eta)^\sharp:Tgr(m)\to Ann(Tgr(m))\otimes \RR^2$ is not isomorphism.%by dimension: $\w^\sharp:Tgr(m)\to Ann(Tgr(m))$ is isomorphism, hence $(\w,\eta)^\sharp$ is not.
For the unit map we can do a similar construction. Let $(\gpd\rightrightarrows M,\w)$ be a symplectic groupoid and $\mathcal{H}\rightrightarrows M$ another Lie groupoid. By the construction in Proposition~\ref{prop:prod} we can prove that the product groupoid $\gpd\times \mathcal{H}$ is a poly-symplectic groupoid with $\RR^2$-valued 2-form $(\w,0)$. It is clear that the graph of the unit map in $\gpd\times \mathcal{H}$ is isotropic. Now take $v\in T_m\mathcal{H}\backslash T_mM$ and verify that $(\w,0)(X+Y,Z+v)=0$ for any $X,Y,Z\in T_mM$. This last claim means that the graph of the unit map in $\gpd\times \mathcal{H}$ is not coisotropic, thus it is neither Lagrangian nor poly-Lagrangian.

\begin{remark}
%\textcolor{red}{Rewrite this in a more understandable way}
In Example~\ref{ex:productofcoadjoint}, we obtained the case of two different $r$-Poisson structures with the same leaf space, since the foliations, given by product of coadjoint orbits, coincide. Moreover, in the previous examples we get that its integrations are not isomorphic. Indeed, we will show that they are not even Morita equivalent. For this note that the normal representation of  the integration of the first $r$-Poisson structure $\gpd$, is $N\gpd=T\G\times \overset{r}{\cdots}\times T\G$  whereas for the second $r$-Poisson structure the normal representation is $N\Theta=T\G$ for the corresponding  integration. It follows that these two integrations are not Morita equivalent as poly-symplectic groupoids(see e.g. the construction in Sec.3.4 and Thm.4.3.1 in \cite{dH}). 
\end{remark}

\section{The sigma model}
This section is divided into two parts, the first one will describe the classical Poisson Sigma Model (PSM) and the construction of the integrating groupoid of a Poisson manifold. In the second part we extend this construction to the case of $r$-Poisson structures, adapting the sigma model conveniently so we obtain PPSM as a particular case of $r$-Poisson structures.

\subsection{The PSM}

Here we give a quick review on the subject making clear the key points in order to extend to our case of interest. For more detail see \cite{CF,ConTh}.

From the AKSZ construction, and considering only fields with ghost number 0 \footnote{For more details on the superfield interpretation of PSM, see e.g. \cite{AKSZ}.}, it
follows that the space of bulk fields of this theory is the space of maps \footnote{vector bundle morphisms.} $\mathcal F=\mbox{Map}(T\Sigma, T^*M)$, whereas the space of boundary fields is $\mathcal F_{\partial}=\mbox{Map}(T\partial \Sigma, T^*M)$. An element $\phi_{\partial} \in\mathcal F_{\partial} $ is parametrized by the tuple $(X,\eta)$ where 

\begin{enumerate}
\item $X \in \mbox{Map}(\partial \Sigma, M)$ (boundary position fields),
\item $\eta \in \Gamma(\mbox{Hom}(T\partial \Sigma, X^*(T^*M)))$ (boundary momentum fields).
\end{enumerate}

The space $F_{\partial}$ can be naturally equipped with a weak symplectic structure $w$ given by 
\begin{equation}
\omega \left (  (X, \eta),(\tilde X, \tilde{\eta}) \right )= \int_0^1 (X\tilde{\eta}-\tilde X \eta)(t) \, dt.
\end{equation}
 
 The PSM action is a functional on $\mathcal F$, where $\phi \in \mathcal F$ is  given in coordinates $X \in \mbox{Map}(\Sigma, M)$ (position fields) and $\eta \in \Gamma(\mbox{Hom}(T\Sigma, X^*(T^*M)))$ (momentum fields). In this coordinates, the action is written as:
\begin{equation}
S(\phi)=\int_{\Sigma} \langle \eta, dX \rangle+ \frac 1 2 \langle  \Pi^{\sharp}(X)\eta, \eta \rangle.
\end{equation}
Restricted to the boundary, this functional has the same form:

\begin{equation}
S_{\partial}(\phi_{\partial})=\int_0^1\langle \eta, dX \rangle+ \frac 1 2 \langle  \Pi^{\sharp}(X)\eta, \eta \rangle.
\end{equation}
The kinetical term $\int_{\Sigma} \langle \eta, dX \rangle$ arises from the AKSZ construction for the source data, and it corresponds to the action functional for abelian 2-dimensional BF theory. The potential term $\int_{0}^1\frac 1 2 \langle  \Pi^{\sharp}(X)\eta, \eta \rangle$ arises from transgression of the Hamiltonian function $\frac 1 2 \langle  \Pi^{\sharp}(X)\eta, \eta\rangle$ into $F_{\partial}$.

Solving the variational problem $\delta S=0$ on the disk with vanishing boundary conditions \cite{CF} we obtain the Euler-Lagrange space 
$\mathcal C= \mbox{Mor}(T\Sigma, T^*M),$
that is, the space of Lie algebroid morphisms between $T\Sigma$ and $T^*M$.
Similarly, restricted to the boundary we obtain
$\mathcal C_{\partial}=\mbox{Mor}(T\partial \Sigma, T^*M)$. 
In \cite{CF} we can find the following theorem

\begin{theorem}\label{Gpd}
$C_{\partial}$ is a coisotropic Banach submanifold of $\mbox{Map}(T[0,1], T^*M)$ and it is the space of cotangent paths, i.e. Lie algebroid maps between $T[0,1]$ and $T^*M$.
\end{theorem}
In the case where $\Sigma$ is a disk with two disjoint vanishing boundary sectors (see \cite{CF, ConTh}) it is proven that
\begin{theorem} The following data
\begin{eqnarray*}
G_0&=& M \\
G_1&=& \underline{C_{\partial}}\\
G_2&=&\{ [X_1, \eta_1], [X_2, \eta_2] \vert X_1(1)=X_2(0)\}
\\
m&\colon & G_2 \to G:= ([X_1, \eta_1], [X_2, \eta_2]) \mapsto [(X_1* X_2, \eta_1* \eta_2)] \\
\varepsilon&:& G_0 \to G_1:= x\mapsto [X\equiv x, \eta\equiv 0] \\
s&\colon &G_1 \to G_0:= [X, \eta]\mapsto X(0) \\
t&\colon &G_1 \to G_0:= [X, \eta]\mapsto X(1) \\
\iota&:& G_1 \to G_1:= [X, \eta] \to [i^* \circ X,i^* \circ \eta]\\
&\mbox{                   }& i\colon [0,1]\to [0, 1]:= t\to 1-t, 
\end{eqnarray*}
correspond to a symplectic groupoid that integrates the Lie algebroid $T^*M$. \footnote{here $*$ denotes path concatenation.}
\end{theorem}
The symplectic form is the one arising from symplectic reduction and it is multiplicative due the fact that groupoid multiplication comes from path concatenation and the additivity property of the integral defining the symplectic form on $\mathcal F_{\partial}$. For further details see \cite{CF}.
This construction is also expressed as the Marsden-Weinstein reduction of the Hamiltonian action  of the (infinite dimensional) Lie algebra
$P_0\Omega^1(M):= \{\beta \in  \mathcal C^{k+1}(I, \Omega^1(M)) \mid \beta(0)=\beta(1)=0\}$
with Lie bracket
\begin{equation}
[\beta,\gamma](u)=d\langle \beta(u), \Pi^{\sharp} \gamma(u) \rangle- \iota_{\Pi^{\sharp}(\beta(u))} d\gamma(u)+ \iota_{\Pi^{\sharp}(\gamma(u))} d\beta(u)
\end{equation}
on the space $T^*(PM)$, on which the equivariant moment map $ \mu\colon  T^*(PM) \to P_0\Omega^1(M)^*$  is described by the equation 
\begin{equation}
\langle \mu(X,\eta), \beta \rangle= \int_0^1 \langle dX(u)+ \Pi^{\sharp}(X(u))\eta(u), \beta(X(u),u) \rangle du.
\end{equation}\\
It turns out that $C_{\Pi}$ is the preimage of $0$, under the equivariant moment map $\mu$, thus $C_{\Pi}$ is coisotropic. In addition to it, the characteristic distribution $TC_{\Pi}^{\omega}$ has finite codimension, more precisely, $2 \dim (M)$, and it is given in local coordinates by the following
\begin{eqnarray*}
\delta_{\beta}X^i(t)&=&-\pi^{ij}(X(t))\beta_j(t)\\
\delta_{\beta}\eta_i(t)&=& d_{t}\beta_i(t) + \partial_i\pi^{jk}(X(t))\eta_j(t)\beta_k(t).
\end{eqnarray*}

%\todo[inline]{Ivan: Could you write this part as is done in your papers and thesis? following our discussion on the PPSM case, I think we must emphasis in the following part of the construction:}

%\subsubsection*{The key points of the construction}
%\begin{itemize}
%\item Definition of the source and target spaces and the weak-symplectic structure on $\F:=\pth(\TM)$ 
%\begin{equation}\label{eq:wsymp}
%\w=\int_0^1
%\end{equation}
%\item The Poisson structure induces a Lie algebra $\LG_0$ that acts on $\F$.
%\item the $\LG_0$-action on $\F$ that is a lifting of a $\LG_0$-action on $\pth(M)$
%\item The momentum map $H$ and its equivariance
%\item $C_\pi$ is coisotropic and $TC_\pi^\w$ is finite generated (which guarantee the MW reduction)
%\item Cosntruction of the groupoid i) the reduction ii)the structural maps iii) multiplicativity of reduced 2-form $\w_r$.
%\end{itemize}

\subsection{The PPSM}
Now we are prepared to introduce a two dimensional TFT such that its target encodes a poly-Poisson structure and its reduced phase space recovers the poly-symplectic integration.

Following the construction of the previous section we consider $\Sigma$ to be $D^2$, the two dimensional disk, such that its boundary $\partial \Sigma$ consists of four intervals $I_0,I_1,I_2,I_3$ under the condition that $I_j$ and $I_{j+1}$ (mod 4) intersect in exactly one point, and both $I_0$ and $I_2$ are vanishing boundary sectors. Given a poly-Poisson structure, the target space of the associated Poly-Poisson sigma model (PPSM) is the bundle $(S,P)$. In this case, the space of bulk fields for PPSM is given by
\begin{equation} \F^{PP}=\mbox{Map}(T\Sigma, T^*M\otimes \mathbb R^r) 
\end{equation}
whereas the space of boundary fields is 
\begin{equation}
\F^{PP}_{\partial}=\mbox{Map}(T\partial \Sigma, T^*M \otimes \mathbb R^r)
\end{equation}
Using the same boundary conditions as in the PSM, we can identify $\F^{PP}_{\partial}$ with the $r$-th Whithey sum of the cotangent bundles of the path-space of $M$. Therefore
\[\F^{PP}_{\partial}:=\oplus_r T^*(\pth(M))\cong \pth(\oplus_r T^*M).\]

\begin{proposition}
$\F^{PP}_{\partial}$ is a weak poly-symplectic Banach manifold.
\end{proposition} 
Before we sketch the proof we should extended the definition of weak-symplectic from to the realm of $r$-symplectic manifolds. In a Banach manifold $B$, a weak $r$-symplectic structure is a closed skew-symmetric bilinear form $\w:TB\times TB\to \RR^r$ so that the induced map $\w^\sharp:TB\to TB^*\otimes \RR^r$ is an injective bundle map.
\begin{proof}
The direct proof is to show that 
\begin{align*}
\bar{\w}_\gamma(\delta(X_1,\eta_2),\delta(X_1,\eta_2))&=
(\dots,\int_0^1\w_{can}(\delta(X_1,\eta_2),\delta(X_1,\eta_2))(t)dt ,\dots )_{i=1,\dots,r}\\
&=(\dots,\int_0^1(\delta_1X_1 \delta_2\eta_2 -\delta_1X_2 \delta_2\eta_1)dt ,\dots )_{i=1,\dots,r}
\end{align*}
is a weak $r$-symplectic structure on $\F^{PP}_{\partial}$. But a more convenient proof is to show that this can be reach as the weak $r$-symplectic product structure. For this recall the definition in \eqref{eq:can.PS} and note that the canonical projection $pr_j:\F^{PP}_{\partial}\to T^*\pth(M)$ satisfies $\cap dpr_j=\{0\}$ that leads us to get the injectivity of $\bar{\w}^\sharp$.
\end{proof}

Now we will consider the Lie algebra of gauge symmetries for the PPSM. Let $\mathfrak g_0$ be the (infinite dimensional) Lie algebra of $\mathcal C^1$-maps $\alpha: [0,1] \to \Gamma(S)$ such that
\begin{enumerate}
\item $\alpha(0)=\alpha(1)=0, \forall \alpha \in \Gamma(S)$.
\item The Lie bracket is given by $ [ \alpha, \beta](t)= \lfloor \alpha(t),\beta(t) \rfloor$, i.e. the bracket is defined pointwise by using the bracket $\lfloor \cdot, \cdot \rfloor$ from the sections on $S$.
\end{enumerate}
By using the $r$-Poisson structure of $(S,P)$ in the definition of $\LG_0$ we can induce two operations:
\begin{align}\label{def:g-PPaction} 
\xi:\LG_0\to \mathfrak{X}(\F^{PP}_{\partial}); \hspace{1cm} 
&\xi_{\beta}(X)=-P_X(\beta)\\ \notag
&\xi_{\beta}(\eta) = d\beta + \partial P_x\eta \beta\\ \label{def:PPm.m.}
H:\LG_0\times \F^{PP}_{\partial}\to \RR^r ; \hspace{1cm}
&(\beta,(X,\eta))\mapsto \int_0^1\beta(dX) dt-\int_0^1\eta(P_X(\beta))dt
\end{align}

\begin{proposition}
Let $\beta$ be an element of $\mathfrak g_0$.
\begin{enumerate}
\item The action by $\beta$ is a lift of an action by $\beta$ on $\mbox{Map} ([0,1], M)$.
\item The action is Hamiltonian weak $r$-symplectic with the Hamiltonian function $H$.% are Hamiltonian with respect to the weak  polysymplectic structure on $F^{PPSM}_{\partial}$.
\end{enumerate}
\end{proposition}

\begin{proof}
The fact that the vector field $\xi_{\beta}$ in the $X$-direction is $\eta$-independent implies that the action by $\mathfrak g_0$ is a lift of the action given by \ref{def:g-PPaction}.

In order to prove the second statement, we denote $f$ as the first integral in the definition of $H$ in \eqref{def:PPm.m.} and $J=-\int_0^1\eta(P_X(\beta))dt$. In the same way as in \cite{CF} we get that $f[\beta,\alpha]=\beta f(\alpha)$\footnote{Note that this relation can also be proved by seen $f$ as a topological moment map, i.e. the fact that is moment map is independent of the $r$-Poisson structure}. Also note that $J(\beta,(X,\eta))=\int_0^1\eta(\xi_\beta X)$ which is moment map for the $\LG_0$-action. This claim follows from the fact that that $\bar{\w}=-d\bar{\alpha}$ is exact and the action $\xi$ on $\F^{PP}_{\partial}$ is a lifting of an action on $\pth(M)$. These two facts show that $H=f+J$ is a moment map.%for a complete proof give a look to SG notes page GS-37
\end{proof}
Note that the argument in Example~\ref{ex:covelocities reduction} can be extended to the infinite dimensional case and we obtain a reducible weak $r$-symplectic action, hence the weak $r$-symplectic reduction at 0-level set can be done. 

In order to obtain a finite-dimensional reduction with 0-level  we need to prove the following result:

\begin{proposition}\label{prop:PS and Coisotropic sbmnfd}
The submanifold $\F(S,P)=\pth(S)\overset{\iota}{\to}\F^{PP}_{\partial}$ is weak $r$-symplectic submanifold with the 2-form $\w=\iota^*\bar{\w}$. Moreover, 
\begin{equation}\label{eq:CSP}
C_{(S,P)}=Mor(T\Sigma,S)=\{\hat{X}:=(X,\eta):X:\Sigma\to M,\eta\in \Gamma(T^*\Sigma,X^*S), dX+P_X(\eta)=0\}
\end{equation}
is coisotropic submanifold with finite codimension.
\end{proposition}

%========================proof of prop===============================
\begin{proof}
For any element $\gamma\in \pth(S)$ we denote by $\gamma_b$ its base path and $(\gamma^1,\dots,\gamma^r)\in S_{\gamma_b}$. If $\bar{\gamma},\bar{\eta}$ are two curves in $\pth(S)$, then 
$$\w(\delta \bar{\gamma},\delta \bar{\eta})=\int_0^1(\dots,\delta_1\bar{\gamma}_b\delta\bar{\eta}^i-\delta_1\bar{\eta}_b\delta_2
\bar{\gamma}^i,\dots)_{i=1,\dots,r}dt.$$
We now condsider that $\delta \bar{\gamma}\in \Ker \w$, and in addition we assume that $\bar{\eta}$ has constant base path and is linear on the fibers, we obtain that
$$\int_0^1(\delta_1\bar{\gamma}_b\bar{\eta}^1,\dots,\delta_1\bar{\gamma}_b\bar{\eta}^r)dt=0$$
for any $(\bar{\eta}^1,\dots,\bar{\eta}^r)\in S_{\bar{\eta}_b}$, hence $\delta_1\bar{\gamma}_b\in S_{\bar{\eta}_b}^0=\{0\}$. This leads us to note that, for any curve $\bar{\eta}$ in $\pth(S)$ we get
$$\w(\delta \bar{\gamma},\delta\bar{\eta})=\int_0^1(-\delta_1\bar{\eta}_b\delta_2\bar{\gamma}^1,\dots,-\delta_1\bar{\eta}_b\delta_2\bar{\gamma}^r)dt=0.$$
As $\delta_1\bar{\eta}$ runs over $TM$, we can conclude that $\delta_2\bar{\gamma}^i=0$, which complete the assertion that $\delta\bar{\gamma}=0$ when $\delta\bar{\gamma}\in \Ker\w$.

For the claim that $C_{(S,P)}$ is coisotropic, we just adapt the same ideas as \cite[Thm.3.1]{CF} and note that the ODE's solving the initial valued problem are in our context vector valued ODE's which obey the same theory of existence and uniqueness of solutions.
\end{proof}
\begin{remark}
Note that the proof of the first claim does not use the existence of the bundle map $P$, so $\pth(S)$ is weak $r$-symplectic for any sub-bundle $S\leq \opr\TM$ so that $S^o=\{0\}$.
\end{remark}
%========================================================================
Let us mention two important facts related to a given weak $r$-symplectic structure  $(\F(S,P),\w)$: 
\begin{enumerate}
\item The $\LG_0$-action on $\F^{PP}_{\partial}$ restricts to $\F(S,P)$.
 \item The map
\begin{align}\label{eq:PSm.m in S}
H:\LG_0\times \F(S,P)\to \RR^r \cong \RR^{*r} ; \hspace{1cm}
&(\beta,(X,\eta))\mapsto \int_0^1\beta(dX) +\beta(P_X(\eta))dt
\end{align}
is a moment map for the restricted action. 
\end{enumerate}
%%=============idea of the proofs===============
%For this just note that $H=\tilde{H}\circ \iota$ and that $\iota$ is $\LG_0$-equivariant
%%===========================================
As a direct result we get the reducibility condition of the action, hence the weak-Marsden-Weinstein reduction $\underline{C_{(S,P)}}$ from $H$ at 0-level.

 Furthermore, as we also have the coisotropic reduction of finte codimension we obtain
\begin{theorem}\label{thm:PS-reduction of CSP}
If $(S,P)$ is integrable, then $\underline{C_{(S,P)}}$ is $r$-symplectic manifold. Moreover, it is equipped with a Lie groupoid structure that makes $\gpd=\underline{C_{(S,P)}}\rightrightarrows M$ a $r$-symplectic groupoid integrating $(S,P)$.
\end{theorem}

\begin{proof}
It was already proved that $\underline{C_{(S,P)}}$ is $r$-symplectic. It remains to prove that is a Lie groupoid, but this is just an extension of the proofs in \cite{CF} with the structural map defined in the same way as in Theorem~\ref{Gpd}.

%\todo[inline]{Ivan: could you help me with this part?}
To end the proof, note that by the transgresion argument of the AKSZ construction,  the $r$-symplectic form $\Omega$ is multiplicative since it is compatible with path concatenation. By construction, we obtain that 
\[\ker(ds)\vert_{\varepsilon(M)}\cong S\]
and the kernel of the target map $t$ restricted to $\varepsilon (M)$ coincides with the map $P$. 

\end{proof}

\subsection{Examples of PPSM integration}
Here we present some construction of the PPSM in the special cases of trivial, $r$-symplectic, product and linear $r$-Poisson structures. The constructions follow same lines in \cite[Sec.~5]{CF}.

\begin{example}[The $r$-symplectic case:Example~\ref{ex:integration-PS}]
Note that for $(X,\eta)\in C_{(S,P)}$ we get that $\eta=-P_X^{-1}(dX)$. Hence, in the same way as in Sec.~5.2 of \cite{CF}, the groupoid $\gpd$ that integrates $(S,P)$ is the fundamental groupoid of $M$, and in the simply connected case is $M\times M$, with the induced $r$-symplectic form using the source and target maps. 

\end{example} 

\begin{example}[The trivial case: Example~\ref{ex:integration-trivial}]
Note that the space in \eqref{eq:CSP} consists on bundle maps with constant path $X:I\to M$ and continuous map $\eta:I\to S_X$. Moreover, the map
\begin{align*}
j:& \gpd\longrightarrow S\\
&(X,\eta)\mapsto (X(0),\int_I\eta dt)
\end{align*}
is an isomorphism. The only thing that we must to verify is that the map is well defined, but this follows from the fact that $X$ is constant, so $\eta$ and its Riemann sums over $I$ belong to the fiber $S_X$. The claim about the isomorphism follows in the same way as in \cite[Sec.5.1]{CF}. In conclusion, for a subbundle $S\leq \opr \TM$ with trivial anchor $P$, its integration is $\gpd=S\rightrightarrows M$ with $r$-symplectic form the restriction of the canonical one on $\opr \TM$.

\end{example} 

\begin{example}[The product case:Example~\ref{ex:integration-product}]
Here we consider two integrable $r$-Poisson structures $(M_i,S_i,P_i)_{i=1,2}$. By the product structure we can verify that the space $C_{(S,P)}$ is the space of solutions of the equation
$$X+P_X\eta=(X_1,X_2)+(P_1(\eta_1),P_2(\eta_2))=0;$$
that is, the composition of solutions $C_{(S_1,P_1)}$ and $C_{(S_2,P_2)}$. This remark leads us to obtain the isomorphism 
\begin{align*}
j:& \gpd\longrightarrow \gpd_1\times \gpd_2\\
&\hat{X}\mapsto (\hat{X_1},\hat{X_2})
\end{align*}
and conclude the integration as in Example~\ref{ex:integration-product}.

\end{example} 

As direct consequence of the previous example we get the PPSM integration of the constant structure as in Example~\ref{ex:integration-constant}.

\begin{example}[The linear case:Example~\ref{ex:integration-linear}]
Recall the definition of the linear $r$-Poisson structure on $\LG^*$ the dual of a finite dimensional Lie algebra $\LG$ in Example~\ref{ex:linear}. From that definition we can note that $S_\zeta\simeq \LG$. Using this identification and the trivialization $\G\times \LG^*_{(r)}\simeq \opr T^*\G$ we can define the following map:
\begin{align*}
j&: \gpd\longrightarrow \opr T^*\G\\
(X&,\eta)\mapsto (Hol(\eta),X(0)).
\end{align*}
The proof that the map $j$ is an isomorphism follows from the (adapted) proof of \cite[Theorem~5.2]{CF}.
\end{example} 

\begin{remark}\label{rmk:cotangent gpd}
The same arguments in Example~\ref{ex:linear} work when we begin with a Lie algebroid $A$ instead of a Lie algebra and we produce a poly--Poisson structure on $\oplus_{(r)} A^*$. The respective vector bundle \eqref{eq:bundle-r-linear} and the anchor map \eqref{eq:anchor-r-linear} are defined by using the Poisson structure of $A^*$ (see for example \cite{We}). The path-space integration of the poly--Poisson manifold $\oplus_{(r)} A^*$ is done in the same spirit of the previous example and the main result of \cite{C}. This procedure yields the integrating poly--symplectic groupoid $\oplus_{(r)}T^*\mathcal{G}\rightrightarrows \prod_{r}A^*$ for $\mathcal{G}\rightrightarrows M$ the source-simply connected Lie groupoid integrating $A\to M$.
\end{remark}

\section{Relational groupoids}
In this section we give a brief exposition of relational groupoids equipped with symplectic structure, which naturally appeared in the context of PSM before gauge reduction \cite{CC, ConTh}.
\begin{definition} \label{Rel}
A \textbf{relational symplectic groupoid} is a triple $(\mathcal G,\, L,\, I)$ where 
\begin{enumerate}
 \item $\mathcal G$ is a weak symplectic manifold. \footnote{In the infinite dimensional setting we restrict to the case of Banach manifolds.} 
\item $L$ is an immersed Lagrangian submanifold of $\mathcal G ^3.$
 \item $I$ is an antisymplectomorphism of $\mathcal G$ called the \emph{inversion},
\end{enumerate}

satisfying the six compatibility axioms A.1-A.6:
\begin{itemize}
\item The cyclicity axiom  (A.1) encodes the cyclic behaviour of the multiplication and inversion maps for groups, namely, if $a,b,c$ are elements of a group $G$ with unit $e$ such that $abc=e$, then $ab=c^{-1}, \, bc= a^{-1}, ca=b^ {-1}$.%The horizontal edges of the boundary of each surface represent the relational symplectic groupoid $\mathcal G$  and the interior of each surface represents the canonical relation.
\item (A.2) encodes the involutivity property of the inversion map of a group, i.e. $(g^{-1})^{-1}=g, \forall g \in G$.
\item (A.3) encodes the compatibility between multiplication and inversion:
$$(ab)^{-1}=b^{-1}a^{-1}, \forall a,b \in G.$$
\item (A.4) encodes the associativity of the product: $a(bc)=(ab)c, \forall a,b,c \in G$.
\item (A.5) encodes the property of the unit of a group of being idempotent: $ee=e$.
\item The axiom (A.6)  states an important difference between the construction of relational symplectic groupoids and usual groupoids. The compatibility between the multiplication and the unit is defined up to an equivalence relation, denoted by $L_2$, whereas for groupoids such compatibility is strict; more precisely, for groupoids such equivalence relation is the identity. In addition, the multiplication and the unit are equivalent with respect to $L_2$.
\end{itemize}

\end{definition}

The graphical description of the axioms A.1-A.6 is given in Figure \ref{fig:Ax}. The morphisms $I_{rel}$ and $L_{rel}$ are represented by a twisted stripe and pair of pants respectively, and the induced immersed canonical relations $L_1, L_2$ and $L_3$ are constructed as compositions of $L$ and $I$. As it is shown in the Figure, they should satisfy  the previously defined compatibility axioms.
\newpage
\begin{figure}[h]
\centering%
\center{\includegraphics[scale=0.55]{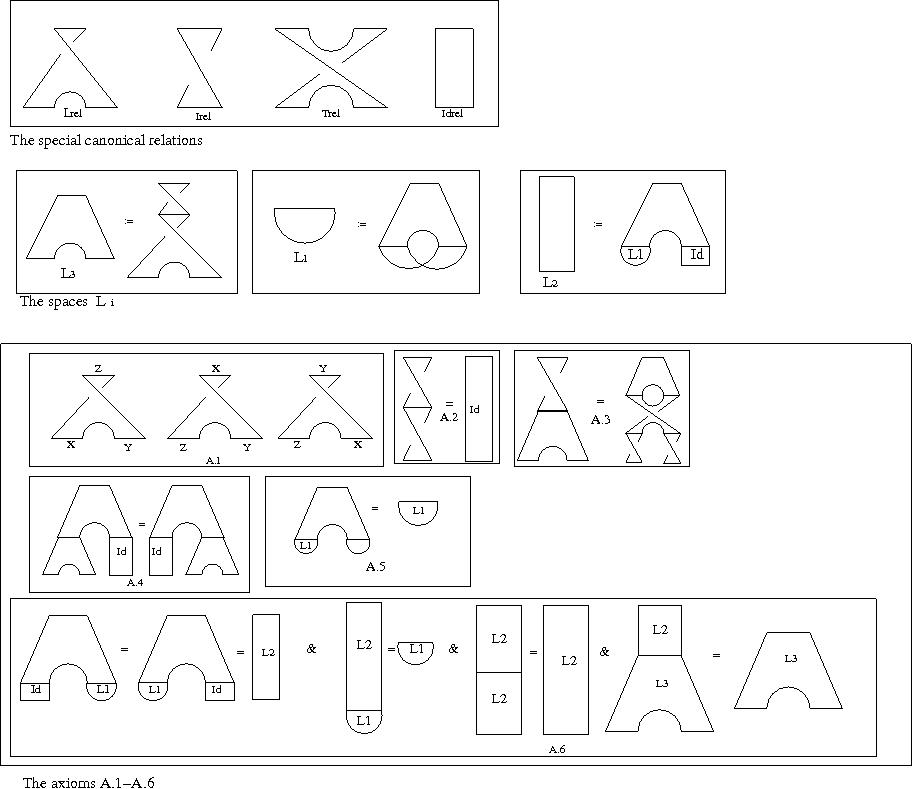}}
\caption{Relative symplectic groupoid: Diagrammatics}
\label{fig:Ax}
\end{figure}

\subsection{Examples}
In particular, symplectic groupoids are particular instances of relational symplectic groupoids, and triples $(\mathcal G, \mathcal L, I)$ where $I$ is an anti-symplectic involution, $L$ is an $I$-invariant Lagrangian in $\mathcal G$ and $L= \mathcal L\times \mathcal L \times  \mathcal L$. More examples of relational groupoids can be found in \cite {CC} and \cite{ConTh}.

This suggests that in the poly-symplectic context, there is a natural groupoid object which serves as integration of poly-Poisson structures, before Marsden-Weinstein reduction.

\begin{definition}\label{RelPS}
A \textbf{relational poly-symplectic groupoid} is a triple $(\mathcal G,\, L,\, I)$ where 
\begin{enumerate}
 \item $\mathcal G$ is a weak poly-symplectic manifold.  
\item $L$ is an immersed Lagrangian submanifold of $\mathcal G ^3.$
 \item $I$ is an antisymplectomorphism of $\mathcal G$
\end{enumerate}
satisfying the six compatibility axioms A.1-A.6:
\end{definition}

\subsection{Lagrangian submanifolds of Poly-symplectic structures}
There are several notions in the literature for Lagrangian submanifolds of poly-symplectic structures. For the purposes of this paper we consider the following definition of Lagrangian.
\begin{definition}
A submanifold $L$ of a poly-symplectic manifold $M$ to be Lagrangian if it is maximally isotropic, i.e. $L=L^{\omega}$. 
\end{definition}
Note that this does not imply that $2\dim (L)=M$ and it does not imply that $L\cong L^0$, as in the usual instance of Lagrangian subspaces of symplectic structures.

\begin{theorem}\label{PPSG}
Let $(S,P)$ be an integrable poly-symplectic structure. Then the evolution relations of the PPSM form a relational poly-symplectic groupoid which integrates $(S,P)$.
\end{theorem}
\begin{proof}
It is easy to observe, as proved in  Example 3.9,  that the graph of the multiplication $gr(\mu)$ of a polysymplectic groupoid is maximally isotropic. Using the Marsden-Weinstein reduction, we observe that the preimage if $gr(\mu)$ under the reduction map is again maximally isotropic, since $\pi^{-1}(gr(\mu)) \subseteq C_{\partial}$. If we set $I$ to be the time reversal antisymplectomorphism for $\mathcal F_{\partial}$, then $L:= I\circ \pi^{-1}(gr(\mu))$ is the immersed Lagrangian which defines the relational groupoid structure for $\mathcal F_{\partial}$.
\end{proof}

\begin{remark}
Note that we do not require all the evolution relations to be Lagrangian. For instance, the evolution relation $L_1$, that is the preimage of the graph of the unit map for the Marsden-Weinstein reduction, is neither Lagrangian nor poly-Lagrangian (see Example 3.9).
\end{remark}

\section{Further Applications}
\subsection{Other integrations}
The path-space integration described above produces the source simply connected poly-symplectic groupoid integrating a poly-Poisson structure $(M,S,P)$. This integration is canonical and \emph{universal}: every Lie groupoid which integrates $(M,S,P)$ is obtained via a discrete quotient of the source simply connected integration. For instance, the holonomy gropoid integrating a Lie algebroid $A \to B$ is a quotient of the Lie groupoid $G\rightrightarrows B$ obtained by the path-space construction, and the quotient is given by holonomy equivalence. This implies that different integrations of the same Poisson manifold are equivalent in the extended symplectic category \cite{CC, ConTh}. In the context of relational integrations of poly-Poisson structures, the path-space integration defined above implies the following result for different integrations of $(M,S,P)$.
\begin{proposition}\label{Class}
Let $G\rightrightarrows M$ be any integration of a poly-Poisson structure $(M,S,P)$ and let $(G, L, I)$ be its corresponding relational poly--symplectic groupoid \footnote{obtained by taking the graphs of the corresponding structure maps of $G\rightrightarrows M$}. Let $(\mathcal G, \mathcal L, I)$ be the relational poly--symplectic groupoid  obtained in Theorem\ref{PPSG}. Then $(G,L,I)$ and  $(\mathcal G, \mathcal L, I)$ are equivalent as relational groupoids.
%Let G ? M and G 0 ? M be two s-fiber connected symplectic groupoid integrating the same Poisson manifold (M, ?). Then (G, L, I) and (G 0 , L0 , I0 ) are equivalent as relational symplectic groupoids. The proof of this theorem is based on the fact that any s- fiber connected integration of a given Lie algebroid A comes from a quotient with respect to the action of a discrete group on the s-fiber simply connected Lie groupoid that integrates A.
\end{proposition}

\begin{proof}
The proof of this proposition is based on the fact that any $s$- fiber connected integration $G\rightrightarrows M$ of $S$ comes from a quotient with respect to the action of a discrete group on the s-fiber simply connected Lie groupoid $G_{ssc}\rightrightarrows M$ that integrates $A$ \cite{HM}. Therefore $G_{ssc}$ is a cover of $G$ and this implies that if $G_{ssc}$ is poly--symplectic, then the quotient map $q: G_{ssc} \to G$ is discrete and thus it induces a poly--symplectic structure on $G$ uniquely.  
 Furthermore, since $q$ is a covering map, the graph  of $q$ is a poly-Lagrangian submanifold of $G_{ssc} \times G$. This follows from the fact that poly-Lagrangianity is preserved after a discrete quotient. Thus $\phi:=gr(q)$ is a morphism in the category of poly--symplectic manifolds. Similarly, the  graph of the Marsden-Weinstein reduction map $Q: (\mathcal G, \mathcal L, I) \to G_{ssc}\rightrightarrows M$ is a poly-Lagrangian submanifold, and thus $\psi:= gr(Q)$ is a well defined morphism.
Following \cite{CC}, it is easy to observe that $\psi$ is an equivalence of relational poly-symplectic groupoids, and therefore, $(G,L,I)$ and  $(\mathcal G, \mathcal L, I)$  are equivalent via $\phi \circ \psi$.
\end{proof}

\subsection{Integration of Lie algebroids via poly-Poisson integration}

A well known fact in Poisson geometry and Lie theory of Lie algebroids is that the a Lie algebroid $A$ is integrable if and only if the Lie algebroid $T^*A^*$ (associated to the Poisson structure in $A^*$) is also integrable. Furthermore, if $\mathcal{G}$ is the source-simply connected Lie groupoid integrating $A$ then the cotangent groupoid $T^*\mathcal{G}\rightrightarrows A^*$ is the source-simply connected symplectic groupoid integrating $T^*A^*$. An analogous result can be proved in the poly--Poisson case following Remark~\ref{rmk:cotangent gpd}.

\begin{proposition}\label{prop:integration A vs +rA*}
A Lie algebroid $A$ is integrable if and only if $\oplus_r A^*$ is integrable as poly--Poisson manifold. 
\end{proposition}
\begin{proof}
First we suppose that $A$ is integrable and consider $\mathcal{G}$ its source-simply connected integration. Following the comment on Remark~\ref{rmk:cotangent gpd} we get $\oplus_rT^*\mathcal{G}\rightrightarrows \oplus_rA^*$ is the source-simply connected poly--symplectic groupoid integrating $\oplus_rA^*$. On the other hand, the natural inclusion of $T^*A^*$ in $S$ defined in \eqref{eq:S in product} covering the inclusion $\iota\colon A^*\to \oplus_r A^*$ on the first component is a Lie algebroid morphism. As it is also known that $A\hookrightarrow T^*A^*$ is Lie algebroid morphism (see \cite{C}), we get that $A$ is Lie subalgebroid of $\oplus_rA^*$ which implies that $A$ is integrable whenever $\oplus_rA^*$ is.
\end{proof}

%=====check this part======
%This fact leads us to obtain integration of the poly--Poisson structure on $\oplus_{(r)} A^*$ via the poly--symplectic groupoid $\oplus_{(r)} T^*\mathcal{G}$. 
%
%In  we obtain the path-sapce integration of $\oplus_{(r)} A^*$.  with object manifold $M$, then the poly--Poisson structure on $\oplus_{(r)} A^*$ is integrable to $\oplus_{(r)}T^*\mathcal{G}$ the poly-symplectic groupoid coming as the product of the symplectic groupoid $T^*\mathcal{G}\rightrightarrows A^*$.
%
%This last remark leads us to study the integration of a Lie algebroid $A$ via the integration of the poly--Poisson manifold $\oplus_{(r)} A^*$ (as in \cite{C}). For this we recall that $A$ is Lie subalgebroid of the Lie algebroid $T^*A^*$ coming from the Poisson structure of $A^*$.

\subsection{Weinstein map and Morita equivalence}
In \cite{Lan}, the author proved a beautiful result relating Morita equivalences in the categories of Lie groupoids and Poisson manifolds by using a canonical construction called the {\it Weinstein map}. Such result can be understood as the classical counterpart of the Muhly--Renault--Williams Theorem for Lie Groupoids. To be more precise, Theorem~3.5 in \cite{Lan} states that if $\mathcal{G}$ and $\mathcal{H}$ are Morita equivalent (as Lie groupoids) then the dual vector bundle of their Lie algebroids, $Lie(\mathcal{G})^*$ and $Lie(\mathcal{H})^*$, are Morita equivalent (as Poisson manifolds). The proof is based on two canonical constructions for Lie groupoids and Lie algebroids, namely the moment map for Lie groupoid actions and the pull-back of a Lie groupoid action (sections 3.2 and 3.4 in \cite{Lan} respectivelly). If the Morita equivalence of $\mathcal{G}$ and $\mathcal{H}$ is realized over $M$, then the mutually orthogonal and complete Poisson maps (realizing the Morita equivalence of Poisson manifold ) 
\begin{equation}\label{eq:PoissonME}
Lie(\mathcal{H})^*\overset{J_R}{\longleftarrow}T^*M\overset{J_L}{\longrightarrow}Lie(\mathcal{G})^*
\end{equation}
are given by the moment map for respective Lie groupoid action (see (3.23) and (3.26) in \cite{Lan}).

The same result, with suitable changes, holds in the realm of poly--Poisson structures. The proof is easily extended just by noticing that the product of $r$ copies of the symplectic moment maps $J_L$ is a poly-symplectic moment map with values in the poly--Poisson manifolds $\oplus_{(r)}Lie(\mathcal{G})^*$ with anchor maps $P_\mathcal{G}$. The same is true for the moment map $J_R$ on the poly--Poisson manifold $\oplus_{(r)}Lie(\mathcal{H})^*$ with anchor map $P_\mathcal{H}$\footnote{The poly--Poisson structure of these two manifolds is explained in Remark~\ref{rmk:cotangent gpd}}. Indeed, we get that $$\bar{J_L}:=\oplus_{(r)}J_L:\oplus_{(r)}T^*M\to \oplus_{(r)}Lie(\mathcal{G})^* \mbox{\ and \ } \bar{J_R}:=\oplus_{(r)}J_R:\oplus_{(r)}T^*M\to \oplus_{(r)}Lie(\mathcal{H})^*$$ satisfy the following conditions:
\begin{itemize}
\item $\bar{J_L}$ and $\bar{J_R}$ are surjective submersion poly--Poisson maps
\item The level sets of the maps are connected and simply connected 
\item The foliations of $\oplus_{(r)}T^*M$ defined by the levels of $\bar{J_L}$ and $\bar{J_R}$ are mutually poly--symplectically orthogonal
\end{itemize}

Now recall the bracket of admissible functions as in Equation~\ref{eq:bracket on functions}. In the particular case of $\oplus_{(r)}T^*M$ we have the explicit formula for the bracket
$$\{(\mathrm{pr}_1^*h_1,\ldots,\mathrm{pr}_r^*h_r),(\mathrm{pr}_1^*g_1,\ldots\mathrm{pr}_r^*g_r)\}=(\mathrm{pr}_1^*\{h_1,g_1\},\ldots,\mathrm{pr}_r^*\{h_r,g_r\})$$
where $\mathrm{pr}_j:\oplus_{(r)}T^*M\to T^*M$ is the $j$-projection. Then we also obtain the following conditions

\begin{itemize}
\item $\{\bar{J_R}^*h,\bar{J_L}^*g\}=0$ for any $h\in C_{\mathrm{adm}}^\infty(\oplus_{(r)}Lie(\mathcal{H})^*)$ and $g\in C_{\mathrm{adm}}^\infty(\oplus_{(r)}Lie(\mathcal{G})^*)$
\item $P_\w(\bar{J_R}^*dh)$ and $P_\w(\bar{J_L}^*dg)$ are complete vector fields if $P_\mathcal{H}(dh)$ and $P_\mathcal{G}(dg)$ are also complete %for $h\in C_{\mathrm{adm}}^\infty(\oplus_{(r)}Lie(\mathcal{H})^*)$ and $g\in C_{\mathrm{adm}}^\infty(\oplus_{(r)}Lie(\mathcal{G})^*)$.
\end{itemize}

Motivated by these results we can extend the definition of Morita equivalence to poly--Poisson manifold. Two $r$-Poisson manifolds $(M,S,P)$ and $(N,\Sigma,\Lambda)$ are {\bf Morita equivalents} if there exists an $r$-symplectic manifold $Q$ and smooth maps 
$$M\overset{\mu}{\longleftarrow}Q\overset{\eta}{\longrightarrow}N$$
so that 
\begin{enumerate}[(1)]
\item $\mu$ and $\eta$ are surjective submersion poly--Poisson maps
\item The level sets of the maps are connected and simply connected 
\item The foliations of $Q$ defined by the levels of $\mu$ and $\eta$ are mutually poly--symplectically orthogonal
\item $\{\mu^*h,\eta^*g\}=0$ for $h$ and $g$ admissible functions on their respective spaces
\item $\mu$ and $\eta$ are poly--Poisson complete maps
\end{enumerate}

In conclusion, based on the previous definition, we have proved the following proposition:

\begin{proposition}\label{prop:poly-Wmap}
If $\mathcal{G}$ and $\mathcal{H}$ are Morita equivalents Lie groupoids, then $\oplus_{(r)}Lie(\mathcal{G})^*$ and $\oplus_{(r)}Lie(\mathcal{H})^*$ are Morita equivalents  poly--Poisson manifolds.% (in the sense of conditions (1)--(5)).
\end{proposition}

%=================bibliography==================


\begin{thebibliography}{14}
\bibitem{AKSZ} M. Alexandrov, M. Kontsevich, A. Schwarz, O. Zaboronsky, \emph{ The geometry of the master equation and 
topological quantum field theories}, Internat. J. Mod. Phys. A12, 1405-1430, 1997.
\bibitem{AC} C. Arias Abad, M. Crainic, \textit{The Weil algebra and the Van Est isomorphism}, Ann. Inst. Fourier (Grenoble) \textbf{61} (2011) 927-970.
\bibitem{Aw} A. Awane, \textit{k-symplectic structures}, J. Math. Phys. \textbf{33} (1992)4046-4052.
%\bibitem{Kontsevich} M. Alexandrov, M. Kontsevich, A. Schwarz, O. Zaboronsky, \emph{ The geometry of the master equation and 
%topological quantum field theories}, Internat. J. Mod. Phys. A12, 1405-1430, 1997.
\bibitem{BC}H. Bursztyn, A. Cabrera, \textit{Multiplicative structure at the infinitesimal level,} Math. Ann. {\bf 353}(2012), 663-705.
\bibitem{Ca} C. Carath\'eodory, {\it Uber die Extremalen und geod\"atischen Felder in der Variationsrechnung der mehrfachen Integrale}, Acta Sci. Math. (Szeged) {\bf 4} (1929), 193--216.
\bibitem{CM} M. Castrill\'on L\'opez, J.E. Marsden, {\it Some remarks on Lagrangian and Poisson reduction for field theories}, Journal of Geometry and Physics {\bf 48}, (2003), 52–83.
\bibitem{C} A. Cattaneo, {\it On the Integration of Poisson Manifolds, Lie Algebroids, and Coisotropic Submanifolds,} Letters in Mathematical Physics, Volume 67, {\bf 1}, (2004), pp 33--48.
\bibitem{CC} A. Cattaneo and I. Contreras, {\sl Relational symplectic groupoids}, Letters in Mathematical Physics, Volume 105, {\bf 5}, (2015), pp 723$-$767.
\bibitem{CF} A.S. Cattaneo and G. Felder, \emph{ Poisson sigma models and symplectic groupoids}, in Quantization of Singular Symplectic Quotients, (ed. N. P. Landsman, M. Pflaum, M. Schlichenmeier), Progress in Mathematics 198 ,Birkh\"auser, 61-93 (2001).
\bibitem{ConTh} I. Contreras, {\it Relational symplectic groupoids and Poisson sigma models with boundary}, PhD Thesis, University of Z\"urich, 2013.
\bibitem{CDW} A. Coste, P. Dazord, A. Weinstein, \textit{Groupo\"ides symplectiques}, Publications du D\'epartement de Math\'ematiques. Nouvelle S\'erie. A, Vol. 2,  i--ii, 1-62, Publ. D\'ep. Math. Nouvelle S\'er. A, 87-2, Univ. Claude-Bernard, Lyon, 1987.
\bibitem{CF1} M. Crainic and R. L. Fernandes. \emph{ Integrability of Lie brackets}, Ann. of Math. 157 {\bf2} (2003), 575-620, . %c%
\bibitem{CF2} M. Crainic and R.L. Fernandes, \emph{Integrability of Poisson brackets}, J. of Differential Geometry 66, (2004) 71-137.
\bibitem{CSS}M. Crainic, M. Salazar, I. Struchiner, \textit{Multiplicative forms and Spencer operators}, Math. Z. (2014).
\bibitem{dD} T. de Donder, {\it Th\'eorie invariante du calcul des variations}, Nuov.  ́\'ed. (Gauthiers--Villars), Paris 1935.
\bibitem{dH} M. del Hoyo, {\it Lie groupoids and their orbispaces,} Portugal. Math. {\bf 70} (2013), 161--209.
\bibitem{FoGo}M. Forger, L. Gomes, {\it Multisymplectic and Polysymplectic structures on fiber bundles,} Rev. Math. Phys. {\bf 25} No.9 (2013).
\bibitem{FPR} M. Forger, C. Paufler, H. R\"omer, {\it The Poisson bracket for Poisson forms in multisymplectic field theory}, Rev. Math. Phys., {\bf 15}(7), (2003), 705--743.
\bibitem{Gu}C. Gunther, \textit{The polysymplectic Hamiltonian formalism in field theory and calculus of variations. I. The local case,} J. Differential Geom. Volume \textbf{25}.
%\bibitem{KS} N.J. Kalton and R.C. Swanson, \emph{A symplectic reflexive Banach space with no Lagrangian subspaces}, Trans. of the AMS, Vol. 273, Number 1 (1982).
\bibitem{He} F. Helein, {\it Hamiltonian formalisms for multidimensional calculus of variations and perturbation theory},
   in Noncompact Problems at the Intersection of Geometry, Analysis, and Topology, Contemporany Mathematics {\bf 350} (2004), 127--149.
\bibitem{IMV}D. Iglesias, J.C Marrero, M. Vaquero, \textit{Poly-Poisson Structures}, Lett. Math. Phys. {\bf 103} (2013) 1103--1133.
\bibitem{Ikeda} N. Ikeda, {\it Two-dimensional gravity and nonlinear gauge theory} , Ann.Phys. 235 (1994) 435--464.
\bibitem{Ka} I. V. Kanatchikov, {\it Geometric (pre)quantization in the polysymplectic approach to field theory,} Differential Geometry and Its Applications {\bf 309}, Proc. Conf., Opava (Czech Republic), Silesian University, Opava, (2001), 309--321.
\bibitem{Lan} N. P. Landsman, {\it The Muhly--Renault--Williams Theorem for Lie Groupoids and its Classical Counterpart,} Letters in Mathematical Physics, Volume 54, (2000), 43--59.
\bibitem{MR-RSV} J.C Marrero, N. Roman-Roy, M. Salgado, S. Vilariño, \textit{Reduction of polysymplectic manifolds,} arXiv:1306.0337.
\bibitem{Martz} N. Martinez, {\it Poly-symplectic groupoids and Poly-Poisson structures,} Letters in Mathematical Physics, Volume 105, {\bf 5}, (2015), 693--721.
\bibitem{NMA} N. Martinez Alba. {\it On higher Poisson and higher Dirac structures,} PhD thesis. Instituto Nacional de Matem ́atica Pura e Aplicada, IMPA, 2015.
\bibitem{McD} D. McDuff, D. Salamon, {\it Introduction to Symplectic Topology,} Oxford Mathematical Monographs, Oxford University Press, New York, 1995.
\bibitem{Strobl} P. Schaller, T. Strobl, {\it Poisson structure induced (topological) field theories}, Modern Phys. Lett. A 9 (1994), no. 33, 3129--3136
%\bibitem{Rudin} W. Rudin, \emph{ Functional Analysis}, Second Ed. Int. Ed. McGraw Hill (1991).

%\bibitem {Strong} J. Kijowski, W.M. Tulczyjew, \emph{A symplectic framework for field theories}, Lect. Notes in Phys. 107, Springer (2014).
%\bibitem{Coisotropic} J. Lorand and A. Weinstein \emph{Decomposition of (co)isotropic relations}, \href{http://arxiv.org/abs/1509.04035}{arxiv.org/abs/1509.04035} (2015).
\bibitem{Welag} A. Weinstein, \emph{Symplectic manifolds and their Lagrangian submanifolds},  Advances in Math. {\bf 6}, (1971), 329--346.
\bibitem{We}A. Weisntein, \textit{Symplectic groupoids and Poisson manifolds}, Bull. American Mathematical Society \textbf{16} (1987), 101-104.
\bibitem{Wey} H. Weyl, {\it Geodesic fields in the calculus of variations}, Ann. Math. (2) {\bf 36} (1935), 607--629.
% \bibitem{Weinstein}A. Weinstein, \emph{ Symplectic Categories}, Proceedings of the Summer School, IST, Lisbon, July 13$-$17, 2009, Port.
%Math. 67 no. 2, 119 (2010).


\end{thebibliography}
\end{document}